\documentclass{llncs}

\usepackage[utf8]{inputenc}

\usepackage[version=0.96]{pgf}
\usepackage{tikz}

\usepackage{amsmath}
\usepackage{amssymb}

\usepackage{caption}
\usepackage{subcaption}
\usepackage{subfig} 

\usetikzlibrary{positioning}
\usetikzlibrary{arrows,automata}
\usetikzlibrary{calc}
\usetikzlibrary{petri}
\usetikzlibrary{decorations.pathreplacing}
\usetikzlibrary{decorations.pathmorphing}
\usetikzlibrary{trees}
\usetikzlibrary{shapes,shapes.geometric,fit}
\usetikzlibrary{backgrounds}

\DeclareMathAlphabet{\marking}{OML}{cmm}{b}{it}

\newcommand{\W}{\mathcal{W}}
\newcommand{\M}{\mathcal{M}}

\newcommand{\R}{\mathbb{R}}

\newcommand {\ra} {\rightarrow}

\newcommand{\mI}{\marking{i}}
\newcommand{\mO}{\marking{o}}

\newcommand{\by}[2][]{\xrightarrow[#1]{#2}}

\newcommand{\prism}{\textsc{Prism}}

\newtheorem{theoremDummy}{Theorem}

\newtheorem{lemmaDummy}{Lemma}



\newcommand{\reductionrule}[4]{
\begin{definition}{#1}\\
\label{#2}
\vspace{3pt}
\setlength{\extrarowheight}{5pt}
\begin{tabularx}{\textwidth}{lX}
{\bfseries Guard}: & {
\begin{minipage}[t]{\linewidth}
#3
\end{minipage}}
\\
{\bfseries Action}: &{
\begin{minipage}[t]{\linewidth}
#4
\end{minipage}}
\end{tabularx}

\end{definition}
}

\usepackage{graphicx}
\usepackage{xcolor}

\usepackage{enumerate}

\usepackage{tabularx}

\usepackage{multirow}

\usepackage{url}
\makeatletter
\g@addto@macro{\UrlBreaks}{\UrlOrds}
\makeatother

\begin{document}
\title{Polynomial Analysis Algorithms for Free Choice Probabilistic Workflow Nets\thanks{This work was partially funded by the DFG Project 5090812 (Negotiations: Ein Modell für nebenläufige Systeme mit niedriger Komplexität).}}
\author{Javier Esparza\inst{1} \and Philipp Hoffmann\inst{1} \and Ratul Saha\inst{2}}
\institute{Technische Universit\"at M\"unchen\and National University of Singapore}
\maketitle\begin{abstract}
We study Probabilistic Workflow Nets (PWNs), a model extending 
van der Aalst's workflow nets 
with probabilities. We give a semantics for PWNs 
in terms of Markov Decision Processes
and introduce a reward model. Using a result by 
Varacca and Nielsen, we show that the expected reward of a 
complete execution of the PWN is independent of the scheduler. 
Extending previous work on reduction of 
non-probabilistic workflow nets, we present reduction rules that 
preserve the expected reward. The rules lead to a polynomial-time 
algorithm in the size of the PWN (not of the Markov decision process) 
for the computation of the expected reward. In contrast, since the 
Markov decision process of PWN can be exponentially larger than the 
PWN itself, all algorithms based on constructing the Markov decision 
process require exponential time. We report on a sample implementation 
and its performance on a collection of benchmarks.
\end{abstract}

\section{Introduction}

Workflow Petri Nets are a class of Petri nets for the representation and analysis of
business processes \cite{DBLP:journals/jcsc/Aalst98,van2004workflow,DBLP:conf/bpm/DeselE00}. 
They are a popular formal back-end for different notations like 
BPMN (Business Process Modeling Notation), EPC (Event-driven Process Chain), 
or UML Activity Diagrams.

There is recent interest in extending these notations, in particular BPMN,
with the concept of cost (see e.g. \cite{magnani2007bpmn,DBLP:conf/icws/SaeediZS10,sampath2011evaluation}). 
The final goal is the development of tool support for computing the worst-case 
or the average cost of a business process. A sound foundation for the latter requires to extend
Petri nets with probabilities and rewards. Since Petri nets can express complex interplay between
nondeterminism and concurrency, the extension is a nontrivial semantic problem which has been studied
in detail (see e.g. \cite{DBLP:journals/tcs/VaraccaVW06,DBLP:journals/tcs/AbbesB08,DBLP:conf/fossacs/AbbesB09}
for untimed probabilistic extensions and \cite{eisentraut2013semantics} for timed extensions).

Fortunately, giving a semantics to probabilistic Petri nets is much simpler for {\em confusion-free}
Petri nets \cite{DBLP:journals/tcs/VaraccaVW06,DBLP:journals/tcs/AbbesB08}, a class 
that already captures many control-flow constructs of BPMN. In particular, 
confusion-free Petri nets strictly contain Workflow Graphs, also called free-choice Workflow Nets \cite{DBLP:journals/jcsc/Aalst98,DBLP:journals/is/FavreFV15,FVM16,esparza2016reduction}.
In this paper we study free choice Workflow Nets extended with rewards and probabilities. 
Rewards are modeled as real numbers attached to the transitions of the workflow,
while, intuitively, probabilities are attached to transitions modeling nondeterministic choices.
Our main result is the first polynomial algorithm for computing the expected reward of a workflow.

In order to define expected rewards, we give untimed, probabilistic confusion-free nets a semantics 
in terms of Markov Decision Processes (MDP), with rewards captured by a reward function. In a nutshell, 
at each reachable marking the enabled transitions are partitioned into {\em clusters}. All transitions 
of a cluster are in conflict, while transitions of different clusters are concurrent. In the MDP
semantics, a scheduler selects one of the clusters, while the transition inside this
cluster is chosen probabilistically. We use MDPs instead of probabilistic event structures, as in 
\cite{DBLP:journals/tcs/VaraccaVW06,DBLP:journals/tcs/AbbesB08,DBLP:conf/fossacs/AbbesB09}, because for our
purposes the semantics are equivalent, and an MDP semantics allows us to use the well established reward 
terminology for MDPs \cite{puterman2014markov}.

In our first contribution, we prove that the expected reward of a confusion-free workflow net 
is independent of the scheduler resolving the nondeterministic choices, and so we can properly speak 
of {\em the} expected reward of a free-choice workflow. The proof relies on a result 
by Varacca and Nielsen \cite{VN} on Mazurkiewicz equivalent schedulers. 

Since MDP semantics of concurrent systems captures all possible interleavings of transitions, the MDP of a free-choice workflow can grow exponentially in 
the size of the net, and so MDP-based algorithms for the expected reward have exponential runtime. 
In our second contribution we provide a polynomial-time {\em reduction algorithm} consisting 
of the repeated application of a 
set of {\em reduction rules} that simplify the workflow while preserving 
its expected reward. Our rules are an extension to the probabilistic case of a set of rules
for free-choice Colored Workflow Nets recently presented in \cite{esparza2016reduction}. 
The rules allow one to merge two alternative tasks, summarize or shortcut two consecutive tasks by one, 
and replace a loop with a probabilistic guard and an exit by a single task. 
We prove that the rules preserve the expected reward. The proof makes crucial use of the fact that the 
expected reward is independent of the scheduler: Given the two workflow nets before and after the reduction, 
we choose suitable schedulers for both of them, and show that the expected rewards under these 
two schedulers coincide. 

Finally, as a third contribution we report on a prototype implementation, and on experimental 
results on a benchmark suite of nearly 1500 workflows derived from industrial business processes. 
We compare our algorithm with the different algorithms based on the construction of the MDP implemented
in {\prism} \cite{KNP11}.




\section{Workflow Nets}
\label{sec:workflowNets}

We recall the definition of a workflow net, and the properties of soundness
and 1-safeness. 

\begin{definition}[Workflow Net \cite{DBLP:journals/jcsc/Aalst98}]
A {\em workflow net} is a tuple $\W=(P,T,F,i,o)$ where
\begin{itemize}
\item $P$ is a finite set of places.
\item $T$ is a finite set of transitions ($P\cap T = \emptyset$).
\item $F \subseteq (P\times T) \cup (T \times P)$ is a set of arcs.
\item $i, o \in P$ are distinguished {\em initial} and {\em final} places such that $i$ has no incoming arcs and $o$ has no outgoing arcs.
\item The graph $(P\cup T, F\cup(o, i))$ is strongly connected.
\end{itemize}
\end{definition}

We write ${}^\bullet p$ and $p^\bullet$ to denote the input and output
transitions of a place $p$, respectively, and similarly ${}^\bullet t$ and 
$t^\bullet$ for the input and output places of a transition $t$.
A {\em marking} $M$ is a function from $P$ to the natural numbers that assigns a 
number of tokens to each place. A transition $t$ is {\em enabled}
at $M$ if all places of ${}^\bullet t$ contain at least one token in $M$. 
An enabled transition may {\em fire}, removing a token 
from each place of ${}^\bullet t$ and adding one token to each 
place of $t^\bullet$. We write $M \by{t}M'$ to denote that $t$ is 
enabled at $M$ and its firing leads to $M'$. The {\em initial marking} ({\em final marking}) 
of a workflow net, denoted by $\mI$ ($\mO$), puts one token on place $i$ (on place $o$), 
and no tokens elsewhere. A sequence of transitions
$\sigma = t_1 \, t_2 \cdots t_n$ is an {\em occurrence sequence} or
{\em firing sequence} if there are markings $M_1, M_2, \ldots, M_n$ such that
$\mI \by{t_1} M_1 \cdots M_{n-1} \by{t_n} M_n$. $\mathit{Fin_\W}$ is the set of all firing sequences of $\W$ that end in the final marking. A marking is {\em reachable} if some 
occurrence sequence ends in that marking. 

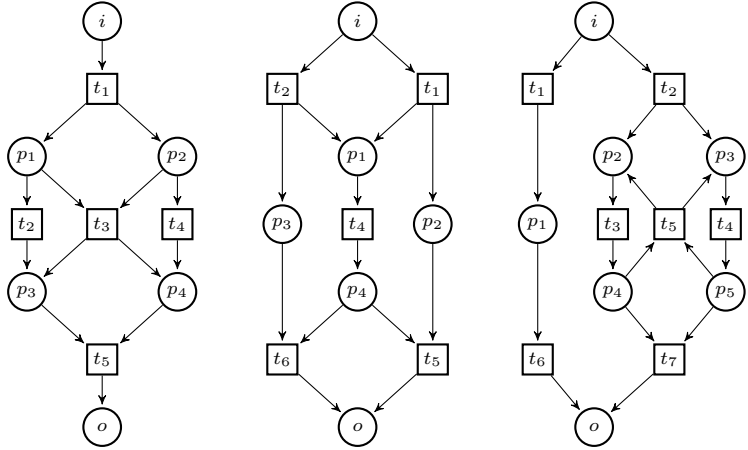
\begin{figure}
\centering
\begin{tikzpicture}[>=stealth',bend angle=45,auto]
	\tikzstyle{every node}=[font=\scriptsize]
	\tikzstyle{place}=[circle,thick,draw=black,fill=white,minimum size=5mm,inner sep=0mm]
	\tikzstyle{transition}=[rectangle,thick,draw=black,fill=white,minimum size=4mm,inner sep=0mm]
	\tikzstyle{every label}=[black]
	
	\node [place] at (0,0) (c0){$i$};
	\node [place] at (-1,-1.8) (c1){$p_1$};
	\node [place] at (1,-1.8) (c2){$p_2$};
        \node [place] at (-1,-3.6) (c3){$p_3$};
	\node [place] at (1,-3.6) (c4){$p_4$};
	\node [place] at (0,-5.4) (o){$o$};

	\node [transition] at (0,-0.9) (t1) {$t_1$}
	edge [pre] (c0)
	edge [post] (c1)
	edge [post] (c2);

	\node [transition] at (-1,-2.7) (t2) {$t_2$}
	edge [pre] (c1)
	edge [post] (c3);

	\node [transition] at (0,-2.7) (t3) {$t_3$}
	edge [pre] (c1)
        edge [pre] (c2)
        edge [post] (c3)
	edge [post] (c4);

        \node [transition] at (1,-2.7) (t4) {$t_4$}
	edge [pre] (c2)
	edge [post] (c4);

       	\node [transition] at (0,-4.5) (t5) {$t_5$}
	edge [pre] (c3)
        edge [pre] (c4)
	edge [post] (o);

\end{tikzpicture}
\qquad
\begin{tikzpicture}[>=stealth',bend angle=45,auto]
	\tikzstyle{every node}=[font=\scriptsize]
	\tikzstyle{place}=[circle,thick,draw=black,fill=white,minimum size=5mm,inner sep=0mm]
	\tikzstyle{transition}=[rectangle,thick,draw=black,fill=white,minimum size=4mm,inner sep=0mm]
	\tikzstyle{every label}=[black]
	
	\node [place] at (0,0) (c0){$i$};
	\node [place] at (0,-1.8) (c1){$p_1$};
	\node [place] at (1,-2.7) (c2){$p_2$};
        \node [place] at (-1,-2.7) (c3){$p_3$};
	\node [place] at (0,-3.6) (c4){$p_4$};
	\node [place] at (0,-5.4) (o){$o$};

	\node [transition] at (1,-0.9) (t1) {$t_1$}
	edge [pre] (c0)
	edge [post] (c1)
	edge [post] (c2);

	\node [transition] at (-1,-0.9) (t2) {$t_2$}
	edge [pre] (c0)
        edge [post] (c1)
	edge [post] (c3);

	\node [transition] at (0,-2.7) (t4) {$t_4$}
	edge [pre] (c1)
	edge [post] (c4);

	\node [transition] at (1,-4.5) (t5) {$t_5$}
	edge [pre] (c2)
        edge [pre] (c4)
	edge [post] (o);

	\node [transition] at (-1,-4.5) (t6) {$t_6$}
	edge [pre] (c3)
	edge [pre] (c4)
	edge [post] (o);

\end{tikzpicture}
\qquad
\begin{tikzpicture}[>=stealth',bend angle=45,auto]
	\tikzstyle{every node}=[font=\scriptsize]
	\tikzstyle{place}=[circle,thick,draw=black,fill=white,minimum size=5mm,inner sep=0mm]
	\tikzstyle{transition}=[rectangle,thick,draw=black,fill=white,minimum size=4mm,inner sep=0mm]
	\tikzstyle{every label}=[black]
	
	\node [place] at (-.75,0) (c0){$i$};
	\node [place] at (-1.5,-2.7) (c1){$p_1$};
	\node [place] at (-0.5,-1.8) (c2){$p_2$};
	\node [place] at (1,-1.8) (c3){$p_3$};
	\node [place] at (-0.5,-3.6) (c4){$p_4$};
	\node [place] at (1,-3.6) (c5){$p_5$};
	\node [place] at (-.75,-5.4) (o){$o$};

	\node [transition] at (-1.5,-0.9) (t1) {$t_1$}
	edge [pre] (c0)
	edge [post] (c1);

	\node [transition] at (0.25,-0.9) (t2) {$t_2$}
	edge [pre] (c0)
	edge [post] (c2)
	edge [post] (c3);

	\node [transition] at (-0.5,-2.7) (t3) {$t_3$}
	edge [pre] (c2)
	edge [post] (c4);

	\node [transition] at (1,-2.7) (t4) {$t_4$}
	edge [pre] (c3)
	edge [post] (c5);

	\node [transition] at (0.25,-2.7) (t5) {$t_5$}
	edge [pre] (c4)
	edge [pre] (c5)
	edge [post] (c2)
	edge [post] (c3);

	\node [transition] at (-1.5,-4.5) (t6) {$t_6$}
	edge [pre] (c1)
	edge [post] (o);

	\node [transition] at (0.25,-4.5) (t7) {$t_7$}
	edge [pre] (c4)
	edge [pre] (c5)
	edge [post] (o);;


\end{tikzpicture}
\caption{Three workflow nets}
\label{fig:example1}
\end{figure}

\begin{definition}[Soundness and 1-safeness \cite{DBLP:journals/jcsc/Aalst98}]
A workflow net is {\em sound} if the final marking is reachable from any reachable marking, and for every transition $t$ there is a reachable marking that enables $t$. A workflow net is {\em 1-safe} if $M(p) \leq 1$ for every reachable marking $M$ and for every  place $p$.
\end{definition}

Figure \ref{fig:example1} shows three sound and 1-safe workflow nets.  
In this paper we only consider 1-safe workflow nets, and identify a marking with the 
set of places that are marked. Markings which only mark a single place are written 
without brackets and in bold, like the initial marking $\mI$. In general, deciding
if a workflow net is sound and 1-safe is a PSPACE-complete problem. However,
for the class of {\em free-choice} workflow nets, introduced below, and 
for which we obtain our main result, there exists a polynomial algorithm \cite{desel2005free}.

\subsection{Confusion-Free and Free-Choice Workflow Nets}

We recall the notions of independent transitions and transitions in conflict.

\begin{definition}[Independent Transitions, Conflict]
Two transitions $t_1$,$t_2$ of a workflow net are {\em independent} if ${}^\bullet t_1 \cap {}^\bullet t_2 = 
\emptyset$. Two transitions are {\em in conflict 
at a marking $M$} if $M$ enables both of them and they are not independent. The set 
of transitions in conflict with a transition $t$ at a marking $M$ is called the {\em conflict set} of $t$ at $M$.
\end{definition}

In Figure \ref{fig:example1} transitions $t_2$ and $t_4$ of the left workflow
are independent, while $t_2$ and $t_3$ are in conflict.
The conflict set of $t_2$ at the marking $\{p_1, p_2\}$ is $\{t_2, t_3\}$, but at the marking $\{p_1, p_4\}$ it is $\{t_2\}$.

It is easy to see that in a 1-safe workflow net two transitions enabled at a marking 
are either independent or in conflict. Assume that a 1-safe workflow net
satisfies the following property: for every reachable marking 
$M$, the conflict relation at $M$ is an equivalence relation. Then, at every reachable marking
$M$ we can partition the set of enabled transitions into equivalence classes, 
where transitions in the same class are in 
conflict and transitions of different classes are independent. For such nets we can introduce
the following simple stochastic semantics: at each reachable marking an equivalence 
 class is selected nondeterministically, and then a transition of the class is 
selected stochastically with probability proportional to a {\em weight} attached to the 
transition. However, not every workflow satisfies this property. For example,
the workflow on the left of Figure  \ref{fig:example1} does not: at the reachable marking 
marking $\{p_1, p_2\}$ transition $t_3$ is in conflict with both $t_2$ and $t_4$, but
$t_2$ and $t_4$ are independent. Confusion-free nets, whose probabilistic semantics is 
studied in \cite{VN}, are a class of nets in which this kind of situation cannot occur. 

\begin{definition}[Confusion-Free Workflow Nets]
A marking $M$ of a workflow net is {\em confused} if there are two independent 
transitions $t_1,t_2$ enabled at $M$ such that $M \by{t_1} M'$ and the conflict sets 
of $t_2$ at $M$ and at $M'$ are different. 
A 1-safe workflow net is {\em confusion-free} if no reachable marking is confused.
\end{definition}

The workflows in the middle and on the right of Figure  \ref{fig:example1} are confusion-free. 

\begin{lemma}[\cite{VN}]
Let $W$ be a 1-safe, confusion-free workflow net. For every 
reachable marking of $W$ the conflict relation on the transitions enabled at $M$ is an equivalence relation.
\end{lemma}

Unfortunately, deciding if a 1-safe workflow net is confusion-free is a PSPACE-complete problem (this can be proved by an easy reduction from the reachability problem for 1-safe Petri nets, see \cite{DBLP:conf/ac/Esparza96} for similar proofs). Free-choice workflow nets are a syntactically defined 
class of confusion-free workflow nets.

\begin{definition}[Free-Choice Workflow Nets  \cite{desel2005free,DBLP:journals/jcsc/Aalst98}]
A workflow net is {\em free-choice} if for every two places $p_1, p_2$ either $p_1^\bullet \cap p_2^\bullet = \emptyset$ or $p_1^\bullet = p_2^\bullet$.
\end{definition}

The workflow in the middle of Figure \ref{fig:example1} is not free-choice, e.g. because of 
the places $p_3$ and $p_4$, but the one on the right is. 

It is easy to see that free-choice workflow nets are confusion-free, but even more: 
in free-choice workflow nets, the conflict set of a transition $t$ is the same at all 
reachable markings that enable $t$. To formulate this, we use the notion of a cluster.

\begin{definition}[Transition clusters]
Let $\W=(P,T,F,i,o)$ be a free-choice workflow net. The {\em cluster} of $t\in T$ is the  set of transitions $[t] = \{ t' \in T \mid {}^\bullet t \cap {}^\bullet t' \neq \emptyset\}$.\footnote{In \cite{desel2005free} clusters are defined in a slightly different way.}
\end{definition}

By the free-choice property, if a marking enables a transition of a cluster, 
then it enables all of them. We say that the marking enables the cluster; we also say that a cluster fires if one of its transitions fires.

\begin{proposition}
\label{prop:fc}
\begin{itemize}
\item Let $t$ be a transition of a free-choice workflow net. For every marking that enables $t$, 
the conflict set of $t$ at $M$ is the cluster $[t]$.
\item Free-choice workflow nets are confusion-free.
\end{itemize}
\end{proposition}
\begin{proof}
The first part follows immediately from the free-choice property.
For the second part, let $t_1,t_2$ be independent transitions enabled at a marking $M$ such that  $M \by{t_1} M'$. By the free-choice property, for every $t \in [t_1]$ the transitions $t$ and $t_2$ are also independent. So the conflict sets of $t_1$ at $M$ and $M'$ are both equal to $[t_1]$.
\qed
\end{proof}

\section{Probabilistic Workflow Nets}
\label{sec:prob}

We introduce Probabilistic Workflow Nets, and give them a 
semantics in terms of Markov Decision Processes. We first recall some basic definitions.

\subsection{Markov Decision Processes}

For a finite set $Q$, let $dist(Q)$ denote the set of probability distributions over $Q$.

\begin{definition}[Markov Decision Process]
A \emph{Markov Decision Process} (MDP) is a tuple $\mathcal{M} = (Q,q_0,{\it Steps})$ where $Q$ is a finite set of states, $q_0\in Q$ is the initial state, and ${\it Steps} \colon Q \ra 2^{dist(Q)}$ is the probability transition function.
\end{definition}

For a state $q$, a probabilistic transition corresponds to first nondeterministically choosing a probability distribution $\mu \in {\it Steps}(q)$ and then choosing the successor state $q'$ probabilistically according to $\mu$.

A path is a finite or infinite non-empty sequence 
$\pi = q_0 \by{\mu_0} q_1 \by{\mu_1} q_2 \ldots$ where $\mu_i \in {\it Steps}(q_i)$ for every $i \geq 0$.
We denote by $\pi(i)$ the $i$-th state along $\pi$ (i.e., the state $q_i$), and by $\pi^i$ the prefix of $\pi$ ending at $\pi(i)$ (if it exists). For a finite path $\pi$, we denote by ${\it last }(\pi)$ 
the last state of $\pi$. A {\em scheduler} is a function that 
maps every finite path $\pi$ of $\mathcal{M}$ to a distribution of ${\it Steps}({\it last}(\pi))$. 

For a given scheduler $S$, let $\mathit{Paths}^S$ denote all infinite paths 
$\pi = q_0 \by{\mu_0} q_1 \by{\mu_1} q_2 \ldots$
starting in $s_0$ and satisfying $\mu_{i} = S(\pi^i)$ for every $i \geq 0$.
We define a probability measure $\mathit{Prob}^S$ on $\mathit{Paths}^S$ in the usual way using 
cylinder sets \cite{kemeny2012denumerable}. 

We introduce the notion of rewards for an MDP.

\begin{definition}[Reward]
\label{def:MDPreward}
A \emph{reward function} for an MDP is a function ${\it rew} \colon S \ra \mathbb{R}_{\geq 0}$.
For a path $\pi$ and a set of states $F$, the reward until $F$ is reached is $$R(F,\pi) := \sum_{i=0}^{\min\{j | \pi(j)\in F\}} rew(\pi(i))$$ if 
the minimum exists, and $\infty$ otherwise. Given a scheduler $S$, the expected reward to reach a set of states $F$ is defined as
	$$E^S(F) := \int_{\pi \in \mathit{Paths}^S} R(F,\pi) \mathrm{d}\mathit{Prob}^S.$$
\end{definition}

\subsection{Syntax and Semantics of Probabilistic Workflow Nets}
\label{subsec:syn}

We introduce Probabilistic Workflow Nets with Rewards, just called Probabilistic Workflow Nets or PWNs in the rest of the paper.

\begin{definition}[Probabilistic Workflow Net with Rewards]
A {\em Probabilistic Workflow Net with Rewards}(PWN) is a tuple $(P,T,F,i,o,w,r)$ where \\$(P,T,F,i,o)$ is a 1-safe confusion-free workflow net, and $w,r \colon T \ra \R^+$ are a {\em weight function} and a {\em reward function}, respectively.
\end{definition}

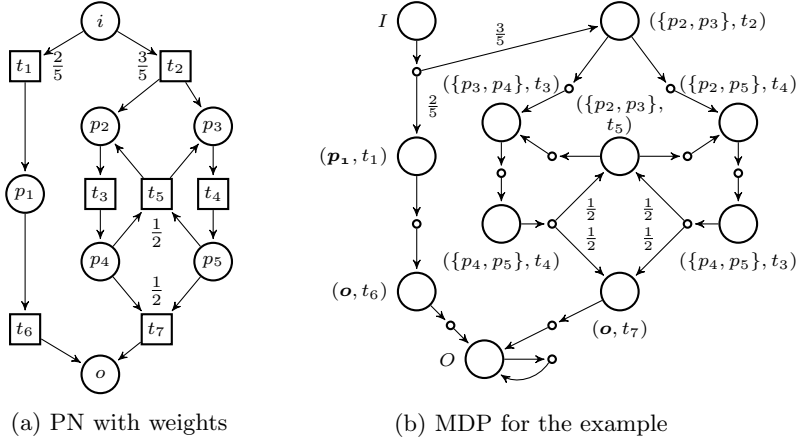
\begin{figure}
\begin{subfigure}[c]{0.4\textwidth}%
\centering%
\begin{tikzpicture}[>=stealth',bend angle=45,auto]
	\tikzstyle{place}=[circle,thick,draw=black,fill=white,minimum size=5mm,inner sep=0mm,font=\scriptsize]
	\tikzstyle{transition}=[rectangle,thick,draw=black,fill=white,minimum size=4mm,inner sep=0mm,font=\scriptsize]
	\tikzstyle{every label}=[black]
	
	\node [place] at (-.5,-0.4) (c0){$i$};
	\node [place] at (-1.5,-2.7) (c1){$p_1$};
	\node [place] at (-0.5,-1.8) (c2){$p_2$};
	\node [place] at (1,-1.8) (c3){$p_3$};
	\node [place] at (-0.5,-3.6) (c4){$p_4$};
	\node [place] at (1,-3.6) (c5){$p_5$};
	\node [place] at (-.5,-5.1) (o){$o$};

	\node [transition] at (-1.5,-1) (t1) {$t_1$}
	edge [pre] (c0)
	edge [post] (c1);

	\node [transition] at (0.5,-1) (t2) {$t_2$}
	edge [pre] (c0)
	edge [post] (c2)
	edge [post] (c3);

	\node [transition] at (-0.5,-2.7) (t3) {$t_3$}
	edge [pre] (c2)
	edge [post] (c4);

	\node [transition] at (1,-2.7) (t4) {$t_4$}
	edge [pre] (c3)
	edge [post] (c5);

	\node [transition] at (0.25,-2.7) (t5) {$t_5$}
	edge [pre] (c4)
	edge [pre] (c5)
	edge [post] (c2)
	edge [post] (c3);

	\node [transition] at (-1.5,-4.5) (t6) {$t_6$}
	edge [pre] (c1)
	edge [post] (o);

	\node [transition] at (0.25,-4.5) (t7) {$t_7$}
	edge [pre] (c4)
	edge [pre] (c5)
	edge [post] (o);;

	\node [right=0cm of t1] {$\frac{2}{5}$};
	\node [left=0cm of t2] {$\frac{3}{5}$};
	\node [below=0cm of t5] {$\frac{1}{2}$};
	\node [above=0cm of t7] {$\frac{1}{2}$};

\end{tikzpicture}%
\caption{PN with weights}%
\label{fig:example2}%
\end{subfigure}
\begin{subfigure}[c]{0.48\textwidth}%
\centering%
\begin{tikzpicture}[>=stealth',scale=0.9]
\tikzstyle{every node}=[font=\scriptsize]
\tikzstyle{place}=[circle,thick,draw=black,fill=white,minimum size=5mm,inner sep=0mm]
\tikzstyle{transition}=[circle,thick,draw=black,fill=white,minimum size=1mm,inner sep=0mm]

\node [place] at (-0,-0) (c0){};
\node[left=0cm of c0]{$I$};
\node [place] at (-0,-2) (c1){};
\node[left=0cm of c1]{$(\marking{p_1},t_1)$};
\node [place] at (-0,-4) (c2){};
\node[left=0cm of c2]{$(\mO,t_6)$};
\node [place] at (3,-0) (c3){};
\node[right=0cm of c3]{$(\{p_2,p_3\},t_2)$};
\node [place] at (4.75,-3) (c4){};
\node[below=0cm of c4]{$(\{p_4,p_5\},t_3)$};
\node [place] at (1.25,-3) (c5){};
\node[below=0cm of c5]{$(\{p_4,p_5\},t_4)$};
\node [place] at (3,-4) (c6){};
\node[below=0cm of c6]{$(\mO,t_7)$};

\node [place] at (4.75,-1.5) (c7){};
\node[above=0cm of c7]{$(\{p_2,p_5\},t_4)$};

\node [place] at (1.25,-1.5) (c8){};
\node[above=0cm of c8]{$(\{p_3,p_4\},t_3)$};

\node [place] at (3,-2) (c9){};
\node[above=0cm of c9,outer sep=-3pt]{\begin{tabular}{c}$(\{p_2,p_3\},$\\$t_5)$\end{tabular}};

\node [place] at (1,-5) (c10){};
\node[left=0cm of c10]{$O$};

\node [transition] at (-0,-0.75) (t1) {}
edge [pre] (c0)
edge [post] node[auto] {$\frac{2}{5}$} (c1)
edge [post] node[auto,outer sep=-2pt] {$\frac{3}{5}$}  (c3);

\node [transition] at (-0,-3) (t2) {}
edge [pre] (c1)
edge [post] (c2);

\node [transition] at (4,-3) (t3) {}
edge [pre] (c4)
edge [post] node[auto,outer sep=-2pt] {$\frac{1}{2}$} (c9)
edge [post] node[auto,swap,outer sep=-2pt] {$\frac{1}{2}$} (c6);

\node [transition] at (2,-3) (t4) {}
edge [pre] (c5)
edge [post] node[auto,swap,outer sep=-2pt] {$\frac{1}{2}$} (c9)
edge [post] node[auto,outer sep=-2pt] {$\frac{1}{2}$} (c6);

\node [transition] at (3.75,-1) (t5) {}
edge [pre] (c3)
edge [post] (c7);

\node [transition] at (2.25,-1) (t6) {}
edge [pre] (c3)
edge [post] (c8);

\node [transition] at (4.75,-2.25) (t7) {}
edge [pre] (c7)
edge [post] (c4);

\node [transition] at (1.25,-2.25) (t8) {}
edge [pre] (c8)
edge [post] (c5);

\node [transition] at (4,-2) (t9) {}
edge [pre] (c9)
edge [post] (c7);

\node [transition] at (2,-2) (t10) {}
edge [pre] (c9)
edge [post] (c8);

\node [transition] at (0.5,-4.5) (t11) {}
edge [pre] (c2)
edge [post] (c10);

\node [transition] at (2,-4.5) (t12) {}
edge [pre] (c6)
edge [post] (c10);

\node [transition] at (2,-5) (t13) {}
edge [pre] (c10)
edge [post,bend left=45] (c10);

\node[inner sep=0,outer sep=0] at (0,-5.5){};

\end{tikzpicture}%
\caption{MDP for the example}%
\label{fig:example2MDP}%
\end{subfigure}
\caption{Running example}
\end{figure}

Figure \ref{fig:example2} shows a free-choice PWN. 
All transitions have reward 1, and so only the weights are represented.
Unlabeled transitions have weight 1.

The semantics of a PWN is an MDP with a reward function. Intuitively, the states of the MDP are pairs $(M, t)$, where $M$ is a marking, and $t$ is the transition that was fired to reach $M$ (since the same marking can be reached by firing different transitions, the MDP can have states 
$(M,t_1)$, $(M,t_2)$ for $t_1 \neq t_2$). Additionally there is a distinguished initial and final states $I, O$. 
The transition relation
${\it Steps}$ is independent of the transition $t$, i.e., 
${\it Steps}((M, t_1)) = {\it Steps}((M,t_2))$ for any two transitions 
$t_1, t_2$, and the reward of a state $(M, t)$ is 
the reward of the transition $t$. Figure \ref{fig:example2MDP} shows
the MDP of the PWN of Figure \ref{fig:example2}, representing 
only the states reachable from the initial state.

\begin{definition}[Probability distribution]
Let $\W=(P, T, F, i, o, w, r)$ be a PWN, let $M$ be a 1-safe marking of $\W$ 
enabling at least one transition, and let $C$ be a conflict set enabled at $M$. 
The {\em probability distribution $P_{M, C}$ over $T$} is obtained by normalizing the weights of the transitions in $C$, and assigning probability $0$ to all other transitions. 
\end{definition}

\begin{definition}[MDP and reward function of a PWN]
\label{def:mdpsem}
Let $\W=(P,T,F,$\\$i,o,w,r)$ be a PWN.
The MDP $M_\W = (Q,q_0,{\it Steps})$ of $\W$ is defined as follows:
\begin{itemize}
\item $Q=({\cal M} \times T) \cup \{I,O\}$ where ${\cal M}$ are the 1-safe markings of 
$\W$, and $q_0 = I$.
\item For every transition $t$: 
\begin{itemize}
\item ${\it Steps}((\mO,t))$ contains exactly one distribution, which
   assigns probability 1 to state $o$, and probability $0$ to all other states.
\item For every marking $M \neq \mO$ enabling no transitions, ${\it Steps}((M,t))$ 
contains exactly one distribution, which assigns probability 1 to $(M,t)$, and probability $0$ to all other states.
\item For every marking $M$ enabling at least one transition,
   ${\it Steps}((M,t))$ contains a distribution $\mu_C$ for each conflict set $C$ of transitions 
   enabled at $M$. The distribution $\mu_C$ is defined as follows. 
   For the states $I, O$: $\mu_C(I)=0=\mu_C(O)$. For each state $(M', t')$ such that $t' \in C$ and $M\by{t'}M'$: $\mu_C((M',t'))=P_{M,C}(t')$. For all other states $(M', t')$: $\mu_C((M',t'))=0$. 
\item ${\it Steps}(I) = {\it Steps}((\mI, t))$ for any transition $t$.
\item ${\it Steps}(O) = {\it Steps}((\mO,t))$ for any transition $t$. 
\end{itemize}
\end{itemize}
The {\em reward function} ${\it rew}_\W$ of $\W$ is defined by: ${\it rew}_\W(I) = 0 = {\it rew}_\W(O)$,
and ${\it rew}_\W((M, t)) = r(t)$.
\end{definition}

In Figure \ref{fig:example2}, ${\it Steps}(i)$ is a singleton set that contains the probability distribution which assigns probability $\frac{2}{5}$ to the state $(\marking{p_1},t_1)$ and probability $\frac{3}{5}$ to the state $(\{p_2,p_3\},t_2)$. ${\it Steps}((\{p_2,p_3\},t_2))$ contains two probability distributions, one that assigns probability $1$ to $(\{p_5,p_3\},t_4)$ and one that assigns probability $1$ to $(\{p_2,p_6\},t_4)$.

We establish a correspondence between firing sequences and paths of the MDP.

\begin{definition}
Let $\W$ be a PWN, and let $M_\W$ be its associated MDP.
Let $\sigma = t_1 t_2\ldots t_n$ be a firing sequence of $\W$. The path $\Pi(\sigma)$ of  $M_\W$ 
corresponding to $\sigma$ is 
$\pi_\sigma = I \by{\mu_0}(M_1,t_1)\by{\mu_1}(M_2,t_2)\by{\mu_2}\ldots$, where $M_0 = \mI$ and 
for every $1 \leq k$:
\begin{itemize}
 \item $M_k$ is the marking reached by firing $t_1 \ldots t_{k}$ from $\mI$, and
 \item $\mu_k$ is the unique distribution of ${\it Steps}(M_{k-1}, t_{k-1})$ such that $\mu(t_k)>0$.
 \end{itemize}
 Let $\pi = I  \by{\mu_0} (M_1, t_1) \cdots (M_n, t_n)$ be a path of $M_\W$. 
The sequence $\Sigma(\pi)$ corresponding to $\pi$ is  $\sigma_\pi = t_1 \ldots t_n$.
 \end{definition}
 
It follows immediately from the definition of $M_\W$ that the functions $\Pi$ and $\Sigma$ are inverses of each other. For a path $\pi$ of the MDP that ends in state $last(\pi)$, the distributions in ${\it Steps}(last(\pi))$ are obtained from the conflict sets enabled after $\Sigma(\pi)$ has fired, if any. If no conflict set is enabled the choice is always trivial by construction. Therefore, a scheduler of the MDP $\M_W$ can be equivalently defined as a function that 
assigns to each firing sequence $\sigma \in T^*$ one of the conflict sets enabled after $\sigma$ has fired. 
In our example, after $t_2$ fires, the conflict sets $\{t_3\}$ and $\{t_4\}$ are concurrently enabled. 
A scheduler chooses either $\{t_3\}$ or $\{t_4\}$. A possible scheduler always chooses 
$\{t_3\}$ every time the marking $\{p_2, p_3\}$ is reached, and produces sequences in 
which $t_3$ always occurs before $t_4$, while others may behave differently.

\medskip
\noindent{\bf Convention:} In the rest of the paper we define schedulers as functions
from firing sequences to conflict sets.
\medskip

In particular, this definition allows us to define the {\em probabilistic language} of a scheduler 
as the function that assigns to each finite firing sequence $\sigma$ the probability of the cylinder of all paths that 
``follow'' $\sigma$. Formally: 

\begin{definition}[Probabilistic language of a scheduler \cite{VN}] 
	The {\em probabilistic language $\nu_S$} of a scheduler $S$
    is the function 
    $\nu_S \colon T^* \rightarrow \R^+$ 
    defined by $\nu_S(\sigma) = \mathit{Prob}^S(cyl^S(\Pi(\sigma)))$.
    A transition sequence $\sigma$ is \emph{produced} by $S$ if $\nu_S(\sigma)>0$.
\end{definition}

The reward function $r$ extends to transition sequences in the natural way by taking the sum of all rewards. 
When we draw a PWN, the labels of transitions have the form $(w,c)$ where $w$ is 
the weight and $c$ is the reward of the transition. See for example Figure \ref{fig:summary:shortcut1}.

We now introduce the expected reward of a PWN under a scheduler.

\begin{definition}[Expected reward of a PWN under a scheduler]
Let $\W$ be a PWN, and let $S$ be a scheduler of its MDP $M_\W$. The expected reward 
$V^S(\W)$ of $\W$ under $S$ is the expected reward $E^S(O)$ to reach the final state $O$ of $M_\W$ .
\end{definition}


%

Given a firing sequence $\sigma$, we have $r(\sigma)=R(O, \Pi(\sigma))$ by the definition of the reward function 
and the fact that $O$ can only occur at the very end of $\pi_\sigma$.

\begin{lemma}
\label{lem:valueSound}
Let $\W$ be a sound PWN, and let $S$ be a scheduler. Then $V^S(\W)$ is finite and 
$V^S(\W)=\sum_{\pi\in\Pi}R(O,\pi)\cdot \mathit{Prob}^S(cyl^S(\pi)) = \sum_{\sigma\in \mathit{Fin_\W}}r(\sigma)\cdot \nu_S(\sigma)$, where $\Pi_{O}$ are the paths of the MDP $M_\W$ leading from the initial state $I$ to the state $O$ (without looping in $O$).
\end{lemma}

\begin{proof}
	By definition, $V^S(\W)=E^S(O)=\int_{\pi \in \mathit{Paths}^S} R(O,\pi) \mathrm{d}\mathit{Prob}^S.$
	Since $\W$ is sound, the final marking is reachable from every marking. Furthermore, since the weights are all positive, and the marking graph is finite, the probability to reach the final marking from any given marking can be bounded away from zero. Therefore the probability to eventually reach the final marking is equal to one, and so $O$ is the only absorbing state of the Markov chain induced by the scheduler $S$.
	It thus holds that
	$$\int_{\pi \in \mathit{Paths}^S} R(O,\pi) \mathrm{d}\mathit{Prob}^S=\int_{\pi \in cyl^S(\Pi_{o})} R(O,\pi) \mathrm{d}\mathit{Prob}^S.$$
	Furthermore, for a path $\pi\in\Pi_{O}$, it holds that $R(O,\pi) = R(O,\pi')$ for all $\pi'\in cyl^S(\pi)$ because $last(\pi)=O$. We obtain
	$$\int_{\pi \in cyl^S(\Pi_{O})} R(O,\pi) \mathrm{d}\mathit{Prob}^S = \sum_{\pi \in \Pi_{O}} R(O,\pi)\cdot \mathit{Prob}^S(cyl^S(\pi))$$ and therefore the first equality.
	Together with $r(\sigma)=R(O,\Pi(\sigma))$, the fact that $\Pi$ is a bijection between $\Pi_O$ and $Fin_\W$, and the definition of $\nu_S$, the second equality follows.
	\qed
\end{proof}

\subsection{Expected Reward of a PWN}

Using a result by Varacca and Nielsen \cite{VN}, we prove 
that the expected reward of a PWN is the same for all schedulers,
which allows us to speak of ``the'' expected reward of a PWN.
We first define partial schedulers.

\begin{definition}[Partial schedulers]
	A {\em partial scheduler} of length $n$ is the restriction of a scheduler 
	to firing sequences of length less than $n$. Given two partial schedulers $S_1, S_2$ of lengths $n_{S_1}, n_{S_2}$, we say
	that $S_1$ {\em extends} $S_2$ if $n_{S_1} \geq n_{S_2}$ and $S_2$ is the restriction of $S_1$ to firing sequences of length 
	less than $n_{S_2}$.
    The {\em probabilistic language $\nu_S$} of a partial scheduler $S$ of length $n$
    is the function $\nu_S \colon T^{\leq n} \rightarrow \R^+$ 
    defined by $\nu_S(\sigma) = \mathit{Prob}^S(cyl^S(\Pi(\sigma)))$.
	A transition sequence $\sigma$ is \emph{produced} by $S$ if $\nu_S(\sigma)>0$.
\end{definition}

Observe that if $\sigma$ is not a firing sequence, then $\nu_S(\sigma)=0$ for every scheduler $S$.
In our running example there are exactly two partial schedulers $S_1, S_2$ of length 2; after $t_2$ they choose
$t_3$ or $t_4$, respectively:
$$\begin{array}{ll}
S_1 \colon & \epsilon \mapsto \{t_1, t_2\} \quad t_1 \mapsto\{t_6\} \quad t_2 \mapsto \{t_3\}
\\
S_2 \colon & \epsilon \mapsto \{t_1, t_2\} \quad t_1 \mapsto\{t_6\} \quad t_2 \mapsto \{t_4\}
\\
\end{array}
$$
For example we have $\nu_{S_1}(t_2t_3)= 3/5$, and $\nu_{S_2}(t_2t_3)= 0$.

For finite transition sequences, Mazurkiewicz equivalence, denoted by $\equiv$, is the smallest 
congruence such that $\sigma t_1 t_2 \sigma' \equiv \sigma t_2 t_1 \sigma'$ for every $\sigma, \sigma' \in T^*$ and for any two {\em independent} transitions $t_1, t_2$ \cite{mazurkiewicz1986trace} . We extend Mazurkiewicz equivalence to 
partial schedulers.

\begin{definition}[Mazurkiewicz equivalence of partial schedulers]
\label{def:mazEquiv}
Given a partial scheduler $S$ of length $n$, we denote by $F_S$ the set of firing 
sequences $\sigma$ of $\W$ produced by $S$ such that either $|\sigma|=n$ or $\sigma$ leads to 
a marking that enables no transitions. 

Two partial schedulers $S_1,S_2$ with probabilistic languages $\nu_{S_1}$ and $\nu_{S_2}$
are {\em Mazurkiewicz equivalent}, denoted $S_1 \equiv S_2$, if they have the same length and there is a bijection 
$\phi \colon F_{S_1} \rightarrow F_{S_2}$ such that $\sigma \equiv \phi(\sigma)$ and 
$\nu_{S_1}(\sigma) = \nu_{S_2}(\phi(\sigma))$ for every $\sigma \in F_n$.
\end{definition}

The two partial schedulers of our running example are not Mazurkiewicz equivalent. Indeed, we 
have $F_{S_1} = \{ t_1t_6, t_2t_3\}$ and $F_{S_2} = \{ t_1t_6, t_2t_4\}$, and no bijection satisfies 
$\sigma \equiv \phi(\sigma)$ for every $\sigma \in F_{S_1}$.

We can now present the main result of \cite{VN}, in our 
terminology and for PWNs.\footnote{In \cite{VN}, enabled conflict sets are called actions,
and markings are called cases.}

\begin{theorem}[Equivalent extension of schedulers \cite{VN}\footnote{
	Stated as Theorem 2, the original paper gives this theorem with $S_1'$ and $S_2'$ being (non-partial) schedulers. However, in the paper equivalence is only defined for partial schedulers and the schedulers constructed in the proof are also partial. 
}]
\label{thm:schedulerEquiv}
Let $S_1$, $S_2$ be two partial schedulers. There exist two partial schedulers $S_1'$, $S_2'$ such that $S_1'$ extends $S_1$, $S_2'$ extends $S_2$ and $S_1'\equiv S_2'$.
\end{theorem}

In our example, $S_1$ can be extended to $S_1'$ by adding $t_1 t_6 \mapsto \emptyset$ and 
$t_2t_3 \mapsto t_4$, and $S_2$ to $S_2'$ by adding  $t_1 t_6 \mapsto \emptyset$ and 
$t_2t_4 \mapsto t_3$. Now we have $F_{S_1'} = \{t_1t_6,t_2t_3t_4\}$ and $F_{S_2'} = \{t_1t_6,t_2t_4t_3\}$. The 
obvious bijection shows $S_1'\equiv S_2'$, because we have $t_2t_3t_4 \equiv t_2 t_4 t_3$ and $\nu_{S_1'}(t_2t_3t_4)= 3/5 =\nu_{S_2}(t_2t_4t_3)$.

We now prove that the expected reward of a PWN is independent of the scheduler.
We need a preliminary  proposition, which follows immediately from the definition of Mazurkiewicz 
equivalence and the commutativity of addition.

\begin{proposition}
	\label{prop:equivValue}
	Let $\W$ be a PWN.
	Then for any two firing sequences $\sigma$ and $\tau$ that are Mazurkiewicz equivalent, it holds that $r(\sigma)=r(\tau)$.
\end{proposition}

\begin{theorem}
	\label{thm:expectedCost}
	Let $\W$ be a PWN. There exists a value $v$ such that for every scheduler $S$ of $M_\W$, 
	the expected reward $V^S(\W)$ is equal to $v$.
\end{theorem}

\begin{proof}
	Pick any two schedulers $R$, $S$. We show that there is a bijection between Mazurkiewicz equivalent firing sequences that end in the final marking and that are produced by those schedulers.

	By Theorem \ref{thm:schedulerEquiv}, any two partial schedulers can be extended to two equivalent partial schedulers, in particular the partial schedulers $R^k$, $S^k$ that are the restrictions of $R$ and $S$ to firing sequences of length less than $k$.

	Let $R'$ be a partial scheduler extending $R^k$, $S'$ a partial scheduler extending $S^k$ such that $R'\equiv S'$. Let $\sigma$ be a firing sequence of length $k$ produced by $R
	$ that ends in the final marking. By the definition of equivalence, there is a firing sequence $\tau$ such that $\sigma\equiv\tau$ and $\nu_{R'}(\sigma)=\nu_{S'}(\tau)$. Since $\sigma$ and $\tau$ are Mazurkiewicz equivalent, $\tau$ also ends in the final marking and also has length $k$. Since $\sigma$ is of length $k$, it was already produced by $R^k$ and thus by $R$, and $\tau$ was already produced by $S$.

	Repeating this for every $k$, we can construct a bijection $\phi$ that maps every firing sequence $\sigma$ produced by $R$ that ends in the final marking to a Mazurkiewicz equivalent firing sequence $\phi(\sigma)$ of the same length produced by $S$ that ends in the final marking. 

	Using Proposition \ref{prop:equivValue}, we know that $r(\sigma)=r(\phi(\sigma))$. Now we apply Lemma \ref{lem:valueSound} and get:
	$$V^{R}(\W)=\sum_{\sigma\in\Sigma}r(\sigma)\cdot \nu_{R}(\sigma)=\sum_{\sigma\in\Sigma}r(\phi(\sigma))\cdot \nu_{S}(\phi(\sigma)) =\sum_{\sigma\in\Sigma}r(\sigma)\cdot \nu_{S}(\sigma) = V^{S}(\W)$$ 
\noindent where the third equality is just a reordering of the sum.
	\qed
\end{proof}

\subsection{Free-choice PWNs}
\label{subsec:free-choice}

By Proposition \ref{prop:fc}, in free-choice PWNs the conflict set of a given transition is exactly its cluster, and so its probability is always the same at any reachable marking that enables it. So we can label a transition directly with this probability.

\medskip

\noindent{\bf Convention:} From now on we assume that the weights are normalized for each cluster, i.e. the weights are already a probability distribution.

\medskip

In the next section we present a reduction algorithm that decides if a given free-choice PWN is sound or not, and if sound computes its expected reward. If the PWN is unsound, then we just apply the following lemma:

\begin{lemma}
\label{lem:valueUnsound}
The expected reward of an unsound free-choice PWN is infinite.
\end{lemma}
\begin{proof}
Let $\W=(P,T,F,i,o)$ be an unsound free-choice PWN. Since, by the definition of a workflow net, the graph $(P\cup T, F\cup(o, i))$ is strongly connected, if we add a transition to $\W$ with $o$ as input and $i$ as output transition, 
we obtain a strongly connected and 1-safe free-choice net $N$. 
Since $\W$ is unsound, by Theorem 1 of \cite{DBLP:journals/jcsc/Aalst98} the net $N$ with the marking $M_0$ that puts one token in place $i$
is either non-live or non-bounded,  and so, since $\W$ is 1-safe, it must be non-live. By Theorem 4.31 of \cite{desel2005free}, the net $N$ with $M_0$ as initial marking has a deadlock $M$, which clearly is also a deadlock of $\W$. Let $i \by{\sigma} M$ be an occurrence sequence leading to $M$. Choose a scheduler $S$ such that $\nu_S(\sigma)>0$. We show that 
the expected reward $V^S(\W)$ is infinite which,  by Theorem  \ref{thm:expectedCost}, implies that the expected reward is also infinite.

The cylinder of paths of $M_\W$ that extend the path $\pi_\sigma$ has positive probability and infinite reward
(by the definition of $\M_\W$ this is the cylinder of paths that extend $\pi_\sigma$ by staying in the
state $(M, t)$ forever, where $t$ is the last transition of $\sigma$ (or in state $i$ forever, if
$\sigma = \epsilon$). So the expected reward $V^S(\W)$ is also infinite. 
\qed
\end{proof}

\begin{figure}
\centering%
\begin{tikzpicture}[>=stealth',bend angle=45,auto]
	\tikzstyle{place}=[circle,thick,draw=black,fill=white,minimum size=5mm,inner sep=0mm,font=\scriptsize]
	\tikzstyle{transition}=[rectangle,thick,draw=black,fill=white,minimum size=4mm,inner sep=0mm,font=\scriptsize]
	\tikzstyle{every label}=[black]
	
	\node [place] at (0,0) (c0){$i$};
	\node [place] at (0,-2) (c1){$p_1$};
	\node [place] at (0,-4) (c2){$o$};

	\node [transition] at (0,-1) (t1) {$t_1$}
	edge [pre] (c0)
	edge [post] (c1);

	\node [transition] at (0,-3) (t2) {$t_2$}
	edge [pre] (c1)
	edge [post] (c2);

	\node [transition] at (1,-2) (t3) {$t_3$}
	edge [pre] (c0)
	edge [pre] (c1)
	edge [post] (c2);

	\node [left=0cm of t1] {$1$};
	\node [left=0cm of t2] {$1$};
	\node [right=0cm of t3] {$2$};

\end{tikzpicture}%
\caption{An unsound confusion free PWN}%
\label{fig:unsoundVal}%
\end{figure}
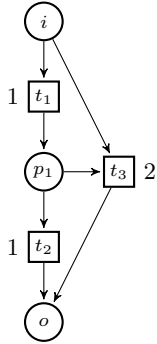

Notice that the above lemma is not true for confusion-free workflow nets, as can be seen in the example net in Figure \ref{fig:unsoundVal}. The transition $t_3$ can never be enabled and thus the net is unsound. However the net contains no deadlock and indeed the only maximal transition sequence is $t_1t_2$. Thus the value of the net is finite.

\section{Reduction rules}
\label{sec:rules}

We transform the reduction rules of  \cite{esparza2016reduction} for non-probabilistic (colored) workflow nets into rules for probabilistic workflow nets. 

\begin{definition}[Rules, correctness, and completeness]
A {\em rule} $R$ is a binary relation on the set of PWNs. We write $\W_1 \by{R} \W_2$ for $(\W_1,\W_2) \in R$. 

A rule $R$ is \emph{correct} if $\W_1 \by{R} \W_2$ implies that $\W_1$ and $\W_2$ are either both sound or both unsound, and have the same expected reward. 

A set $\mathcal{R}$ of rules is \emph{complete} for a class of PWNs if for every sound PWN $\W$ in that class there exists a sequence 
$\W \by{R_1} \W_1 \cdots \W_{n-1}\by{R_n} \W'$ such that $\W'$ 
is a PWN consisting of a single transition $t$ between the two only places $i$ and $o$. 
\end{definition}

Observe that if $\W$ is reduced to a $\W'$ as above, then the expected reward of $\W$ is equal to the reward of $t$ in $\W'$.

As in \cite{esparza2016reduction}, we describe rules as pairs of a \emph{guard} and an \emph{action}. 
$\W_1 \by{R} \W_2$ holds if $\W_1$ satisfies the guard, and $\W_2$ is a 
possible result of applying the action to $\W_1$. 

\paragraph*{Merge rule.} The {\em merge rule} merges two transitions with the same input and output places into one single transition. The weight of the new transition is the sum of the old weights, and the reward is the weighted average of the reward of the two merged transitions.

\reductionrule{Merge rule}{def:merge}
{
$\W$ contains two distinct transitions $t_1, t_2 \in T$ such that ${}^\bullet t_1 = {}^\bullet t_2$ and $t_1^\bullet = t_2^\bullet$.
}{
\vspace{-\topsep}
\begin{enumerate}[(1)]
\item $T := (T\setminus \{t_1,t_2\})\cup \{t_m\}$, where $t_m$ is a fresh name. 
\item $t_m^\bullet := t_1^\bullet$ and ${}^\bullet t_m := {}^\bullet t_1$. 
\item $r(t_m) := w(t_1)\cdot r(t_1)+w(t_2)\cdot r(t_2)$.
\item $w(t_m)=w(t_1) + w(t_2)$.
\end{enumerate}
}

\paragraph*{Iteration rule.} Loosely speaking, the iteration rule removes arbitrary iterations of a transition by adjusting the weights of the possible successor transitions. The probabilities are normalized again and the reward of each successor transition increases by a geometric series dependent on the reward and weight of the removed transition.

\reductionrule{Iteration rule}{def:iteration}
{
$\W$ contains a cluster $c$ with a transition $t \in c$ such that $t^\bullet={}^\bullet t$. 
}{
\vspace{-\topsep}
\begin{enumerate}[(1)]
\item $T := (T \setminus \{t\})$. 
\item For all $t' \in c \setminus\{t\}$: $r(t'):=\frac{w(t)}{1-w(t)}\cdot r(t)+r(t')$
\item For all $t' \in c \setminus\{t\}$: $w(t'):=\frac{w(t')}{1-w(t)}$
\end{enumerate}
}

Observe that $\frac{w(t)}{1-w(t)}\cdot r(t)=(1-w(t))\cdot\sum_{i=0}^\infty w(t)^i\cdot i\cdot r(t)$ captures the fact that $t$ can be executed arbitrarily often, each execution yields the reward $r(t)$, and eventually some other transition occurs.

For an example of an application of the iteration rule, consult Figure \ref{fig:summary:shortcut2} and Figure \ref{fig:summary:shortcut3}. Transition $t_9$ has been removed and as a result the label of transition $t_7$ changed.

\paragraph*{Shortcut rule.} The shortcut rule merges transitions of two clusters into one single transition with the same effect. The reward of the new transition is the sum of the rewards of the old transitions, and its weight the product of the old weights.

A transition $t$ {\em unconditionally enables} a cluster $c$
if ${}^\bullet t'\subseteq t^\bullet$ for some transition $t'\in c$. Observe that if $t$ unconditionally enables $c$ then any marking reached by firing $t$ enables every transition in $c$.

\reductionrule{Shortcut rule}{def:shortcut}
{
$\W$ contains a transition $t$ and a cluster $c\neq[t]$ such that $t$ unconditionally enables $c$.
}{
\vspace{-\topsep}
\begin{enumerate}[(1)]
\item $T := (T \setminus \{t\}) \cup \{t'_s \mid t' \in c\}$, where $t'_s$ are fresh names. 
\item For all $t' \in c$: ${}^\bullet t'_s := {}^\bullet t$ and $t'_s{}^\bullet := (t^\bullet \setminus {}^\bullet t')\cup t'^\bullet$. 
\item For all $t' \in c$: $r(t'_s) := r(t)+r(t')$.
\item For all $t' \in c$: $w(t'_s)=w(t)\cdot w(t')$.
\item If ${}^\bullet p = \emptyset$ for all $p\in c$, then remove $c$ from $\W$.
\end{enumerate}
}

For an example shortcut rule application, compare the example of Figure \ref{fig:example2} with the net in Figure \ref{fig:summary:shortcut1}. The transition $t_1$ which unconditionally enabled the cluster $[t_6]$ has been shortcut, a new transition $t_8$ has been created, and $t_1$, $p_1$ and $t_6$ have been removed.

\begin{theorem}
\label{thm:correctness}
The merge, shortcut and iteration rules are correct for PWNs.
\end{theorem}
\begin{proof}
	It was already shown in \cite{esparza2016reduction} that the rules preserve soundness for free-choice workflow nets. We thus only have to show that the rules preserve the expected reward of the net. In the unsound case this is easy: Since there is a reachable marking from which the final marking is unreachable, there is a cylinder which occurs with positive probability and never reaches the final marking. For such a cylinder, the reward is infinite by Definition \ref{def:MDPreward}, thus the expected reward is infinite. As the rules preserve unsoundness, they also preserve the expected reward in that case.

	By Theorem \ref{thm:expectedCost} the expected reward of the net does not depend on the scheduler. We use this fact in the following way: For each rule, we pick two schedulers, one for the net before the rule application and one for the net after the rule was applied. These schedulers will be such that it is easy to show that their expected rewards are equal. We begin with the shortcut rule.

\paragraph{Shortcut rule.}

	Let $\W_1$, $\W_2$ be such that $\W_1\by{\text{shortcut}} \W_2$.
	Let $c$, $t$ be as in Definition \ref{def:shortcut}. Let $S_1$ be a scheduler for $W_1$ such that $S_1(\sigma_1)=c$ if $\sigma_1$ ends with $t$. Since $t$ unconditionally enables $c$, this is a valid scheduler.

	We define a mapping $\phi$ that maps firing sequences in $\W_2$ to firing sequences in $\W_1$ by replacing every occurrence of $t'_s$ by $t\,t'$. Next we define a scheduler $S_2$ for $\W_2$ by $S_2(\sigma_2) = S_1(\phi(\sigma_2))$.

	Observe that $\phi$ is a bijection between sequences produced by $S_1$ that do not end with $t$ and sequences produced by $S_2$. In particular $\phi$ is a bijection between sequences produced by $S_1$ and $S_2$ that end with the final marking.

	Let now $\sigma_2$ be a firing sequence in $\W_2$ and let $\sigma_1=\phi(\sigma_2)$. We claim that $\sigma_1$ and $\sigma_2$ have the same reward and also $\nu_{S_1}(\sigma_1)=\nu_{S_2}(\sigma_2)$. Indeed, since the only difference is that every occurrence of $t'_s$ is replaced by $t\,t'$ and $r(t'_S)=r(t)+r(t')$ and $w(t'_s)=w(t)w(t')$ by the definition of the shortcut rule, the reward must be equal and $\nu_{S_1}(\sigma_1)=\nu_{S_2}(\sigma_2)$.

	We now use these equalities, the fact that there is a bijection between firing sequences that end with the final marking, and Lemma \ref{lem:valueSound}:

	\begin{eqnarray*}
V(\W_2) & =  & \sum_{\sigma_2 \in \mathit{Fin}_{\W_2}} r(\sigma_2)\cdot \nu_{S_2}(\sigma) = \sum_{\sigma_2 \in \mathit{Fin}_{\W_2}} r(\phi(\sigma_2))\cdot\nu_{S_1}(\phi(\sigma_2)) \\
	& =  & \sum_{\sigma_1 \in \mathit{Fin}_{\W_1}} r(\sigma_1) \cdot \nu_{S_1}(\sigma_1) = V(\W_1) \ .
\end{eqnarray*}

\paragraph{Iteration rule.}
	Let $\W_1$, $\W_2$ be such that $\W_1\by{\text{iteration}} \W_2$.
	Let $c$, $t$ be as in Definition \ref{def:iteration}. Let $S_2$ be a scheduler for $W_2$ such that $S_2(\sigma_2)=c$ if $c$ is enabled after $\sigma_2$. 

	We define a mapping $\phi$ that maps firing sequences in $\W_1$ to firing sequences in $\W_2$ by removing all occurrences of $t$. Next we define a scheduler $S_1$ for $\W_1$ by $S_1(\sigma_1) = S_2(\phi(\sigma_1))$. Note that $\phi$ is not a bijection but it is surjective. 

Let $r_1$ and $r_2$ be the reward functions of $\W_1$ and $\W_2$. For a sequence $\sigma_2$ in $\W_2$, we claim:
$$r_2(\sigma_2)\cdot \nu_{S_2}(\sigma_2) = \sum_{\sigma_1\in\phi^{-1}(\sigma_2)}r_1(\sigma_1)\cdot \nu_{S_1}(\sigma_1) \ .$$

Let $k$ be the number of times $c$ is enabled during $\sigma_2$. We only consider the case $k=1$,
the general case being similar. We observe that $\sigma_2$ is also a sequence in $\W_1$. 
We have 
\begin{eqnarray}
\nu_{S_1}(\sigma_2) & = & \nu_{S_2}(\sigma_2) \cdot(1-w(t)) \label{eq:nu1} \\
r_1(\sigma_2) & = & r_2(\sigma_2)-\frac{w(t)}{1-w(t)}\cdot c(t) \label{eq:r1}
\end{eqnarray}
\noindent because the probabilistic choice must pick something other than $t$, and because 
the iteration rule adds $\frac{w(t)}{1-w(t)}\cdot c(t)$ to the reward of every transition in $c$ in $\W_2$.

	We now insert $l$ occurrences of $t$ in $\sigma_2$, at the position at which $c$ is enabled, and call the new sequence $\tau_l$. We have $\phi^{-1}(\sigma_2) = \{\tau_l \mid l \geq 0 \}$. Further $r_1(\tau_l)=r_1(\sigma_2)+l\cdot c(t)$ and $\nu_{S_1}(\tau_l)=\nu_{S_1}(\tau)\cdot w(t)^l$,
and so summing over all $l$ we get:
$$\begin{array}{rcll}
\displaystyle \sum_{\sigma_1\in\phi^{-1}(\sigma_2)}r_1(\sigma_1)\cdot \nu_{S_1}(\sigma_1) & = & \displaystyle \sum_{l=0}^\infty r_1(\tau_l)\cdot \nu_{S_1}(\tau_l)  \\[0.2cm]
& = & \displaystyle \nu_{S_1}(\sigma_2) \cdot \sum_{l=0}^\infty (r_1(\sigma_2)+l\cdot c(t)) \cdot w(t)^l \\[0.3cm]
& = & \displaystyle  \nu_{S_1}(\sigma_2) \cdot \bigg( \frac{r_1(\sigma_2)}{1-w(t)} +\frac{ c(t) \cdot w(t)}{(1-w(t))^2} \bigg) \\[0.3cm]
& = & \displaystyle  \nu_{S_2}(\sigma_2) \cdot \bigg( r_1(\sigma_2) + \frac{c(t) \cdot w(t)}{1-w(t)} \bigg) & \mbox{(by \ref{eq:nu1})}\\ [0.3cm]
& = & \displaystyle \nu_{S_2}(\sigma_2) \cdot r_2(\sigma_2) & \mbox{(by \ref{eq:r1})}
\end{array}$$
\noindent and the claim is proved.

Now, using the claim we obtain:
	 \begin{eqnarray*}
V(\W_2) & = & \sum_{\sigma_2 \in \mathit{Fin}_{\W_2}} r_2(\sigma_2)\cdot \nu_{S_2}(\sigma) = \sum_{\sigma_2 \in \mathit{Fin}_{\W_2}}\sum_{\sigma_1\in\phi^{-1}(\sigma_2)}r_1(\sigma_1)\cdot\nu_{S_1}(\sigma_1) \\
& = & \sum_{\sigma_1 \in \mathit{Fin}_{\W_1}} r_1(\sigma_1) \cdot \nu_{S_1}(\sigma_1) = V(\W_1)
\end{eqnarray*}
\noindent where the third equality follows from the fact that $\phi$ is defined on all sequences of $\W_1$ and thus $\phi^{-1}$
hits every sequence in $\W_1$ exactly once.

\paragraph{Merge rule.}
	Let $\W_1$, $\W_2$ be such that $\W_1\by{\text{merge}} \W_2$.
	Let $t_1$, $t_2$ be as in Definition \ref{def:merge}. Let $S_2$ be a scheduler for $\W_2$.

	We define a mapping $\phi$ that maps firing sequences in $\W_1$ to firing sequences in $\W_2$ by replacing all occurrences of $t_1$ and $t_2$ by $t_m$. We define a scheduler $S_1$ for $\W_1$ by $S_1(\sigma_1) = S_2(\phi(\sigma_1))$.

	Once again, $\phi$ is a surjective function. For a sequence $\sigma_2$ in $\W_2$, we claim $r(\sigma_2)\cdot \nu_{S_2}(\sigma_2) = \sum_{\sigma_1\in\phi^{-1}(\sigma_2)}r(\sigma_1)\cdot \nu_{S_1}(\sigma_1)$. Indeed, every sequence $\sigma_1$ the set $\phi^{-1}(\sigma_2)$ can be obtained by replacing $t_m$ by either $t_1$ or $t_2$. So, by Definition \ref{def:merge}, the sums are equal.

	As for the iteration rule, this equality and the fact that $\phi$ is defined for every sequence in $\W_1$ imply that 
the expected rewards of $\W_1$ and $\W_2$ are equal.

\qed
\end{proof}

In  \cite{esparza2016reduction} we provide a reduction algorithm for non-probabilistic free-choice workflow, and prove the following result.

\begin{theorem}[Completeness\cite{esparza2016reduction}]
The reduction algorithm summarizes every sound free choice workflow net
in at most $\mathcal{O}(|C|^4\cdot |T|)$ applications of the shortcut rule and $\mathcal{O}(|C|^4+|C|^2\cdot|T|)$ applications of the merge and iteration rules, where $C$ is the set of clusters of the net.
Any unsound free-choice workflow nets can be recognized as unsound in the same number of rule applications.
\end{theorem}

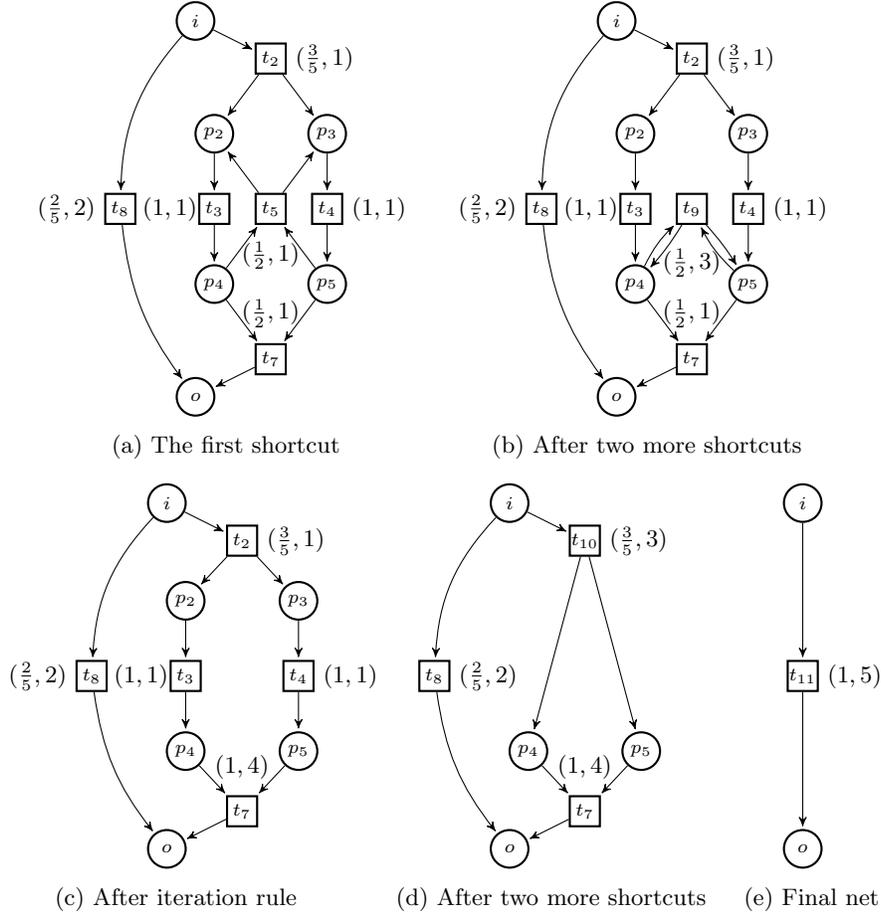
\begin{figure}[tbh]
\centering
\begin{subfigure}[c]{0.45\textwidth}
\centering
	\scalebox{1}{
		\begin{tikzpicture}[>=stealth',bend angle=45,auto]
	\tikzstyle{place}=[circle,thick,draw=black,fill=white,minimum size=5mm,inner sep=0mm,font=\scriptsize]
	\tikzstyle{transition}=[rectangle,thick,draw=black,fill=white,minimum size=4mm,inner sep=0mm,font=\scriptsize]
	\tikzstyle{every label}=[black]

	\node [place] at (-.75,-0.5) (c0){$i$};
	\node [place] at (-0.5,-2) (c2){$p_2$};
	\node [place] at (1,-2) (c3){$p_3$};
	\node [place] at (-0.5,-4) (c4){$p_4$};
	\node [place] at (1,-4) (c5){$p_5$};
	\node [place] at (-.75,-5.5) (o){$o$};

	\node [transition] at (-1.75,-3) (t1) {$t_8$}
	edge [pre,bend left=20] (c0)
	edge [post,bend right=15] (o);

	\node [transition] at (0.25,-1) (t2) {$t_2$}
	edge [pre] (c0)
	edge [post] (c2)
	edge [post] (c3);

	\node [transition] at (-0.5,-3) (t3) {$t_3$}
	edge [pre] (c2)
	edge [post] (c4);

	\node [transition] at (1,-3) (t4) {$t_4$}
	edge [pre] (c3)
	edge [post] (c5);

	\node [transition] at (0.25,-3) (t5) {$t_5$}
	edge [pre] (c4)
	edge [pre] (c5)
	edge [post] (c2)
	edge [post] (c3);

	\node [transition] at (0.25,-5) (t7) {$t_7$}
	edge [pre] (c4)
	edge [pre] (c5)
	edge [post] (o);;

	\node [left=0cm of t1] {$(\frac{2}{5},2)$};
	\node [right=0cm of t2] {$(\frac{3}{5},1)$};
	\node [left=-0.1cm of t3] {$(1,1)$};
	\node [right=0cm of t4] {$(1,1)$};
	\node [below=0.1cm of t5] {$(\frac{1}{2},1)$};
	\node [above=0.1cm of t7] {$(\frac{1}{2},1)$};

\end{tikzpicture}
	}
	\subcaption{The first shortcut}
	\label{fig:summary:shortcut1}
\end{subfigure}
\begin{subfigure}[c]{0.45\textwidth}
\centering
	\scalebox{1}{
		\begin{tikzpicture}[>=stealth',bend angle=45,auto]
	\tikzstyle{place}=[circle,thick,draw=black,fill=white,minimum size=5mm,inner sep=0mm,font=\scriptsize]
	\tikzstyle{transition}=[rectangle,thick,draw=black,fill=white,minimum size=4mm,inner sep=0mm,font=\scriptsize]
	\tikzstyle{every label}=[black]

	\node [place] at (-.75,-0.5) (c0){$i$};
	\node [place] at (-0.5,-2) (c2){$p_2$};
	\node [place] at (1,-2) (c3){$p_3$};
	\node [place] at (-0.5,-4) (c4){$p_4$};
	\node [place] at (1,-4) (c5){$p_5$};
	\node [place] at (-.75,-5.5) (o){$o$};

	\node [transition] at (-1.75,-3) (t1) {$t_8$}
	edge [pre,bend left=20] (c0)
	edge [post,bend right=15] (o);

	\node [transition] at (0.25,-1) (t2) {$t_2$}
	edge [pre] (c0)
	edge [post] (c2)
	edge [post] (c3);

	\node [transition] at (-0.5,-3) (t3) {$t_3$}
	edge [pre] (c2)
	edge [post] (c4);

	\node [transition] at (1,-3) (t4) {$t_4$}
	edge [pre] (c3)
	edge [post] (c5);

	\node [transition] at (0.25,-3) (t9) {$t_9$}
	edge [pre,bend right=10] (c4)
	edge [pre,bend right=10] (c5)
	edge [post,bend left=5] (c4)
	edge [post,bend left=5] (c5);

	\node [transition] at (0.25,-5) (t7) {$t_7$}
	edge [pre] (c4)
	edge [pre] (c5)
	edge [post] (o);
	
	\node [left=0cm of t1] {$(\frac{2}{5},2)$};
	\node [right=0cm of t2] {$(\frac{3}{5},1)$};
	\node [left=-0.1cm of t3] {$(1,1)$};
	\node [right=0cm of t4] {$(1,1)$};
	\node [below=0.2cm of t9] {$(\frac{1}{2},3)$};
	\node [above=0.1cm of t7] {$(\frac{1}{2},1)$};

\end{tikzpicture}
	}
	\subcaption{After two more shortcuts}
	\label{fig:summary:shortcut2}
\end{subfigure}
\begin{subfigure}[c]{0.4\textwidth}
\centering
	\scalebox{1}{
		\begin{tikzpicture}[>=stealth',bend angle=45,auto]
\node[inner sep=0,outer sep=0] at (0,0){};

	\tikzstyle{place}=[circle,thick,draw=black,fill=white,minimum size=5mm,inner sep=0mm,font=\scriptsize]
	\tikzstyle{transition}=[rectangle,thick,draw=black,fill=white,minimum size=4mm,inner sep=0mm,font=\scriptsize]
	\tikzstyle{every label}=[black]

	\node [place] at (-.75,-0.5) (c0){$i$};
	\node [place] at (-0.5,-1.8) (c2){$p_2$};
	\node [place] at (1,-1.8) (c3){$p_3$};
	\node [place] at (-0.5,-3.8) (c4){$p_4$};
	\node [place] at (1,-3.8) (c5){$p_5$};
	\node [place] at (-.75,-5.1) (o){$o$};

	\node [transition] at (-1.75,-2.8) (t1) {$t_8$}
	edge [pre,bend left=20] (c0)
	edge [post,bend right=15] (o);

	\node [transition] at (0.25,-1) (t2) {$t_2$}
	edge [pre] (c0)
	edge [post] (c2)
	edge [post] (c3);

	\node [transition] at (-0.5,-2.8) (t3) {$t_3$}
	edge [pre] (c2)
	edge [post] (c4);

	\node [transition] at (1,-2.8) (t4) {$t_4$}
	edge [pre] (c3)
	edge [post] (c5);

	\node [transition] at (0.25,-4.6) (t7) {$t_7$}
	edge [pre] (c4)
	edge [pre] (c5)
	edge [post] (o);

	\node [left=0cm of t1] {$(\frac{2}{5},2)$};
	\node [right=0cm of t2] {$(\frac{3}{5},1)$};
	\node [left=-0.1cm of t3] {$(1,1)$};
	\node [right=0cm of t4] {$(1,1)$};
	\node [above=0.1cm of t7] {$(1,4)$};

\end{tikzpicture}
	}
	\subcaption{After iteration rule}
	\label{fig:summary:shortcut3}
\end{subfigure}
\begin{subfigure}[c]{0.4\textwidth}
\centering
	\scalebox{1}{
		\begin{tikzpicture}[>=stealth',bend angle=45,auto]
\node[inner sep=0,outer sep=0] at (0,0){};

	\tikzstyle{place}=[circle,thick,draw=black,fill=white,minimum size=5mm,inner sep=0mm,font=\scriptsize]
	\tikzstyle{transition}=[rectangle,thick,draw=black,fill=white,minimum size=4mm,inner sep=0mm,font=\scriptsize]
	\tikzstyle{every label}=[black]
	
	\node [place] at (-.75,-0.5) (c0){$i$};
	\node [place] at (-0.5,-3.8) (c4){$p_4$};
	\node [place] at (1,-3.8) (c5){$p_5$};
	\node [place] at (-.75,-5.1) (o){$o$};

	\node [transition] at (-1.75,-2.8) (t1) {$t_8$}
	edge [pre,bend left=20] (c0)
	edge [post,bend right=15] (o);

	\node [transition] at (0.25,-1) (t10) {$t_{10}$}
	edge [pre] (c0)
	edge [post] (c4)
	edge [post] (c5);

	\node [transition] at (0.25,-4.6) (t7) {$t_7$}
	edge [pre] (c4)
	edge [pre] (c5)
	edge [post] (o);

	\node [right=0cm of t1] {$(\frac{2}{5},2)$};
	\node [right=0cm of t10] {$(\frac{3}{5},3)$};
	\node [above=0.1cm of t7] {$(1,4)$};

\end{tikzpicture}
	}
	\subcaption{After two more shortcuts}
	\label{fig:summary:shortcut4}
\end{subfigure}
\begin{subfigure}[c]{0.15\textwidth}
	\centering
	\scalebox{1}{
		\begin{tikzpicture}[>=stealth',bend angle=45,auto]
\node[inner sep=0,outer sep=0] at (0,0){};
	\tikzstyle{place}=[circle,thick,draw=black,fill=white,minimum size=5mm,inner sep=0mm,font=\scriptsize]
	\tikzstyle{transition}=[rectangle,thick,draw=black,fill=white,minimum size=4mm,inner sep=0mm,font=\scriptsize]
	\tikzstyle{every label}=[black]
	
	\node [place] at (-0,-0.5) (c0){$i$};
	\node [place] at (-0,-5.1) (o){$o$};

	\node [transition] at (-0,-2.8) (t11) {$t_{11}$}
	edge [pre] (c0)
	edge [post] (o);
	\node [right=0cm of t11] {$(1,5)$};
	\node [left=0.25cm of t11] {};
\end{tikzpicture}
	}
	\subcaption{Final net}
	\label{fig:summary:end}
\end{subfigure}
\caption{Example of reduction}
\end{figure}

We illustrate a complete reduction by reducing the example of Figure \ref{fig:example2}. 
We set the reward for each transition to $1$, so the expected reward of the net is the expected number of transition
firings until the final marking is reached. Initially, $t_1$ unconditionally enables $[t_6]$ 
and we apply the shortcut rule. Since $[t_6]=\{t_6\}$, exactly one new transition $t_{8}$ is created. 
Furthermore $t_1$, $p_1$ and $t_6$ are removed (Figure \ref{fig:summary:shortcut1}). Now, 
$t_5$ unconditionally enables $[t_3]$ and $[t_4]$. We apply the shortcut rule twice and call 
the result $t_{9}$ (Figure \ref{fig:summary:shortcut2}). Transition $t_{9}$ now satisfies 
the guard of the iteration rule and can be removed, changing the label of $t_7$ (Figure \ref{fig:summary:shortcut3}).
Since $t_2$ unconditionally enables $[t_3]$ and $[t_4]$, we apply the shortcut rule twice and call the result $t_{10}$
(Figure \ref{fig:summary:shortcut4}). After short-cutting $t_{10}$, we apply the merge rule to
the two remaining transitions, which yields a net with one single transition labeled by $(1,5)$ (Figure \ref{fig:summary:end}). So the net terminates with probability 1 after firing 5 transitions in average.

\paragraph{Fixing a scheduler.} Since the expected reward of a PWN $\W$ is independent of the scheduler, we can fix a scheduler $S$ and compute the expected reward $V^S(\W)$. This requires to compute only
the Markov chain induced by $S$, which can be much smaller than the MDP. However, it is easy to see that this idea does not lead to a polynomial algorithm. Consider the free-choice PWN of Figure \ref{fig:sched}, and the scheduler that always chooses the largest enabled cluster
according to the order 
$$\{t_{11}, t_{12}\} > \cdots > \{t_{n1}, t_{n2}\} > \{u_{11}\} > \{ u_{12}\} > \cdots > \{u_{n1}\} > \{ u_{n2}\}$$
Then for every subset $K \subset \{1, \ldots, n\}$ the Markov chain contains a state enabling $\{ u_{i1} \mid i \in K\} \cup \{u_{i2} \mid i \notin K\}$, and has therefore exponential size. There might be a procedure to find a suitable scheduler for a given PWN such that the Markov chain has polynomial size, but we do not know of such a procedure.

 \begin{figure}[ht]
 	\centering
 	\begin{tikzpicture}[>=stealth',bend angle=45,auto,scale=0.7,every node/.style={scale=0.7}]
	\tikzstyle{every node}=[font=\scriptsize]
	\tikzstyle{place}=[circle,thick,draw=black,fill=white,minimum size=4mm,inner sep=0mm]
	\tikzstyle{transition}=[rectangle,thick,draw=black,fill=white,minimum size=3mm,inner sep=0mm]
	\tikzstyle{every label}=[black]

	\node [place] at (0,0) (c0){$i$};
	\node [place] at (2,-1.5) (c1){};
	\node [place] at (4,-2) (c11){};
	\node [place] at (4,-1) (c12){};
	\node [place] at (2,1.5) (c2){};
	\node [place] at (4,1) (c21){};
	\node [place] at (4,2) (c22){};
	\node [place] at (6,-1.5) (c3){};
	\node [place] at (6,1.5) (c4){};
	\node [place] at (8,0) (o){$o$};

	\node [transition] at (1,0) (t1) {}
	edge [pre] (c0)
	edge [post] (c1)
	edge [post] (c2);

	\node at (4,0) {{\Huge$\ldots$}};

	\node [transition] at (3,-2) (t2) {$t_{n1}$}
	edge [pre] (c1)
	edge [post] (c11);

	\node [transition] at (3,-1) (t3) {$t_{n2}$}
	edge [pre] (c1)
	edge [post] (c12);
	
	\node [transition] at (5,-2) (t2) {$u_{n1}$}
	edge [pre] (c11)
	edge [post] (c3);

	\node [transition] at (5,-1) (t3) {$t_{n2}$}
	edge [pre] (c12)
	edge [post] (c3);


	\node [transition] at (3,2) (t4) {$t_{11}$}
	edge [pre] (c2)
	edge [post] (c22);

	\node [transition] at (3,1) (t5) {$t_{12}$}
	edge [pre] (c2)
	edge [post] (c21);
	
	\node [transition] at (5,2) (t2) {$u_{11}$}
	edge [pre] (c22)
	edge [post] (c4);

	\node [transition] at (5,1) (t3) {$u_{12}$}
	edge [pre] (c21)
	edge [post] (c4);


	\node [transition] at (7,0) (t6) {}
	edge [pre] (c3)
	edge [pre] (c4)
	edge [post] (o);

\end{tikzpicture}
 	\caption{Example}
 	\label{fig:sched}
 \end{figure}

\section{Experimental evaluation}
\label{sec:experiments}

We have implemented our reduction algorithm as an extension of the 
algorithm described in \cite{esparza2016reduction}. In this section we report 
on its performance and on a comparison with \prism{}\cite{KNP11}. The results confirm 
what could be expected:
our polynomial algorithm for free-choice workflows outperforms \prism{}'s 
exponential, but more generally applicable algorithm. More interestingly,
they provide quantitative information on the speed-up achieved by our 
algorithm.

\paragraph{Industrial benchmarks.} The benchmark suite consists of 1385 free-choice workflow nets, 
previously studied in \cite{fahland2009instantaneous}, of which 470 nets are sound. The workflows correspond to 
business models designed at IBM. Since they do not contain probabilistic information, 
we assigned to each transition $t$ the probability $\frac{1}{|[t]|}$ (i.e., the probability is 
distributed uniformly among the transitions of a cluster). We study the following questions, which 
can be answered by both our algorithm and \prism{}: Is the probability to reach the final marking equal to one 
(equivalent to  ``is the net sound?''). And if so, how many transitions must be fired in average to 
reach the final marking? (This corresponds to a reward function assigning reward 1 to each transition.) 

All experiments were carried out on an i7-3820 CPU using 1 GB of memory.

\prism{} has three different analysis engines able to compute expected rewards: 
explicit, sparse and symbolic (bdd). In a preliminary experiment with a timeout of 30 seconds, we observed that the 
explicit engine clearly outperforms the other two: It solved 1309 cases, while the bdd and 
sparse engines only solved 636 and 638 cases, respectively. Moreover, 418 and 423 of the unsolved cases 
were due to memory overflow, so even with a larger timeout the explicit engine is still leading. For this reason, in the comparison we only used the explicit engine.

After increasing the timeout to 10 minutes, the explicit engine did not solve any further case, leaving
76 cases unsolved. This was due to the large state space of the nets: 69 out of the 76 have over  
$10^6$ reachable states. 

The 1309 cases were solved by the explicit engine in 353 seconds, with about 10 seconds for the larger nets. 
Our implementation solved all 1385 cases in 5 seconds combined. It never needs
more than $20$ ms for a single net, even for those with more than $10^7$ states (for these nets 
we do not know the exact number of reachable states). 

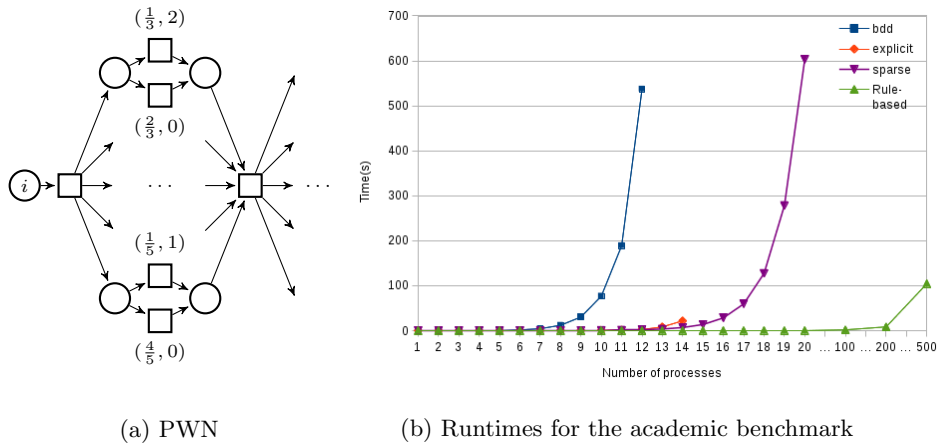
\begin{figure}[t]
        \begin{subfigure}[c]{0.35\textwidth}
	\centering
	\begin{tikzpicture}[>=stealth',bend angle=45,auto,scale=0.6,every node/.style={scale=0.6}]
	\tikzstyle{every node}=[font=\scriptsize]
	\tikzstyle{place}=[circle,thick,draw=black,fill=white,minimum size=4mm,inner sep=0mm]
	\tikzstyle{transition}=[rectangle,thick,draw=black,fill=white,minimum size=3mm,inner sep=0mm]
	\tikzstyle{every label}=[black]

	\node [place] at (0,0) (c0){$i$};
	\node [place] at (2,-2.5) (c1){};
	\node [place] at (2,2.5) (c2){};
	\node [place] at (4,-2.5) (c3){};
	\node [place] at (4,2.5) (c4){};
	\node [place] at (9,0) (o){$o$};

	\node [transition] at (1,0) (t1) {}
	edge [pre] (c0)
	edge [post] (c1)
	edge [post] (2,-1)
	edge [post] (2,0)
	edge [post] (2,1)
	edge [post] (c2);

	\node at (3,0) {$\ldots$};

	\node [transition] at (3,-3) (t2) {}
	edge [pre] (c1)
	edge [post] (c3);

	\node [transition] at (3,-2) (t3) {}
	edge [pre] (c1)
	edge [post] (c3);

	\node [below=0cm of t2] {$(\frac{4}{5},0)$};
	\node [above=0cm of t3] {$(\frac{1}{5},1)$};

	\node [transition] at (3,2) (t4) {}
	edge [pre] (c2)
	edge [post] (c4);

	\node [transition] at (3,3) (t5) {}
	edge [pre] (c2)
	edge [post] (c4);

	\node [below=0cm of t4] {$(\frac{2}{3},0)$};
	\node [above=0cm of t5] {$(\frac{1}{3},2)$};

	\node [transition] at (5,0) (t6) {}
	edge [pre] (c3)
	edge [pre] (4,-1)
	edge [pre] (4,0)
	edge [pre] (4,1)
	edge [pre] (c4)
	edge [post] (6,-2.5)
	edge [post] (6,-1)
	edge [post] (6,0)
	edge [post] (6,1)
	edge [post] (6,2.5);

	\node at (6.5,0) {$\ldots$};

	\node [transition] at (8,0) (t6) {}
	edge [pre] (7,-2.5)
	edge [pre] (7,-1)
	edge [pre] (7,0)
	edge [pre] (7,1)
	edge [pre] (7,2.5)
	edge [post] (o);

\end{tikzpicture}
	\caption{PWN}
	\label{fig:implExample}
        \end{subfigure}
	\begin{subfigure}[c]{0.63\textwidth}
		\centering
		\includegraphics[scale=0.5]{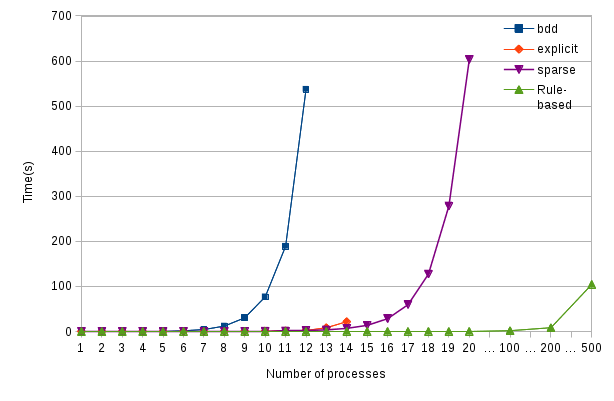}
		\caption{Runtimes for the academic benchmark}
		\label{fig:implTimes}
	\end{subfigure}
	\caption{Academic benchmark}
\end{figure}

In the unsound case, our implementation still reduces the reachable state space by a lot, 
which makes it easier to apply state exploration tools for other problems than the expected reward,
like the distribution of the rewards. After reduction, the 69 nets with at least $10^6$ states had an average 
of 5950 states, with the largest at 313443 reachable states.


\paragraph{An academic benchmark.} Many workflows in our suite have a large state space because of fragments modeling the following situation.
Multiple processes do a computation step in parallel, after which they synchronize. Process $i$ may execute its step normally with probability $p_i$, or 
a failure may occur with probability $1-p_i$, which requires to take a recovery action and therefore has a higher cost. Such a scenario is modeled
by the free-choice PWNs net of Figure \ref{fig:implExample}, where 
the probabilities and costs are chosen at random. The scenario can also be easily modeled
in  \prism{}. Figure \ref{fig:implTimes} shows the time needed by the three
\prism{} engines and by our implementation for computing the expected reward using a time limit of 10 minutes. The number of reachable 
states grows exponentially in the number processes, and the explicit engine runs out of memory
for 15 processes. Since the failure probabilities vary between the processes, there is little structure 
that the symbolic engine can exploit, and it times out for 13 processes. The sparse engine reaches 
the time limit at 20 processes. However, since the rule-based approach does not need to construct the state space, 
we can easily solve the problem with up to 500 processes.

\section{Conclusion}
\label{sec:conclusion}

We have presented a set of reduction rules for probabilistic workflow 
nets with rewards that preserve soundness and the expected reward of the net, 
and are complete for free-choice nets. While the semantics and the expected reward 
are defined via an associated Markov Decision Process, our rules work directly 
on the workflow net. The rules lead to the first polynomial-time algorithm 
to compute the expected reward. 


In future work we want to generalize our algorithm in several ways. First, 
we think that the cost model can be extended to any semiring satisfying some mild conditions. 
A particular instance of this result should lead to an algorithm for 
computing the probability on non-termination and the conditional expected reward under termination, 
which is of interest in the unsound case. Second, we plan to extend our approach to GSPNs with 
the semantics introduced in \cite{eisentraut2013semantics}. Third, we think that the expected {\em time
to termination} of a free-choice workflow can also be computed by means of a reduction algorithm.

\paragraph{Acknowledgments.} We thank the anonymous referees for their comments, 
and especially the one who helped us correct a mistake in Lemmas \ref{lem:valueSound} and \ref{lem:valueUnsound}.

\bibliography{ref}

\begin{thebibliography}{10}

\bibitem{DBLP:journals/jcsc/Aalst98}
W.~v.~d. Aalst.
\newblock The application of {P}etri nets to workflow management.
\newblock {\em Journal of Circuits, Systems, and Computers}, 8(1):21--66, 1998.

\bibitem{van2004workflow}
W.~v.~d. Aalst and K.~M.~v. Hee.
\newblock {\em Workflow management: models, methods, and systems}.
\newblock MIT press, 2004.

\bibitem{DBLP:journals/tcs/AbbesB08}
S.~Abbes and A.~Benveniste.
\newblock True-concurrency probabilistic models: {M}arkov nets and a law of
  large numbers.
\newblock {\em Theor. Comput. Sci.}, 390(2-3):129--170, 2008.

\bibitem{DBLP:conf/fossacs/AbbesB09}
S.~Abbes and A.~Benveniste.
\newblock Concurrency, $\sigma$-algebras, and probabilistic fairness.
\newblock In {\em Proceedings of FOSSACS 2009}, LNCS, vol. 5504, pages
  380--394, 2009.

\bibitem{DBLP:conf/bpm/DeselE00}
J.~Desel and T.~Erwin.
\newblock Modeling, simulation and analysis of business processes.
\newblock In {\em Business Process Management}, LNCS, vol. 1806, pages
  129--141. Springer, 2000.

\bibitem{desel2005free}
J.~Desel and J.~Esparza.
\newblock {\em Free choice Petri nets}, volume~40.
\newblock Cambridge university press, 2005.

\bibitem{eisentraut2013semantics}
C.~Eisentraut, H.~Hermanns, J.-P. Katoen, and L.~Zhang.
\newblock A semantics for every {GSPN}.
\newblock In {\em Application and Theory of Petri Nets and Concurrency}, pages
  90--109. Springer, 2013.

\bibitem{DBLP:conf/ac/Esparza96}
J.~Esparza.
\newblock Decidability and complexity of {P}etri net problems - {A}n
  introduction.
\newblock In {\em Lectures on Petri Nets {I:} Basic Models, Advances in Petri
  Nets}, LNCS, vol. 1491, pages 374--428, 1996.

\bibitem{esparza2016reduction}
J.~Esparza and P.~Hoffmann.
\newblock Reduction rules for colored workflow nets.
\newblock In {\em Proceedings of {FASE} 2016}, LNCS, vol. 9633, pages 342--358,
  2016.

\bibitem{fahland2009instantaneous}
D.~Fahland, C.~Favre, B.~Jobstmann, J.~Koehler, N.~Lohmann, H.~V{\"o}lzer, and
  K.~Wolf.
\newblock Instantaneous soundness checking of industrial business process
  models.
\newblock In {\em Business Process Management}, LNCS, vol. 5701, pages
  278--293. Springer, 2009.

\bibitem{DBLP:journals/is/FavreFV15}
C.~Favre, D.~Fahland, and H.~V{\"{o}}lzer.
\newblock The relationship between workflow graphs and free-choice workflow
  nets.
\newblock {\em Inf. Syst.}, 47:197--219, 2015.

\bibitem{FVM16}
C.~Favre, H.~V\"olzer, and P.~M\"uller.
\newblock Diagnostic information for control-flow analysis of {W}orkflow
  {G}raphs (a.k.a. {F}ree-{C}hoice {W}orkflow nets).
\newblock In {\em Proceedings of {TACAS} 2016}, LNCS, vol. 9636, pages
  463--479, 2016.

\bibitem{kemeny2012denumerable}
J.~G. Kemeny, J.~L. Snell, and A.~W. Knapp.
\newblock {\em Denumerable Markov chains: with a chapter of Markov random
  fields by David Griffeath}, volume~40.
\newblock Springer Science \& Business Media, 2012.

\bibitem{KNP11}
M.~Kwiatkowska, G.~Norman, and D.~Parker.
\newblock {PRISM} 4.0: Verification of probabilistic real-time systems.
\newblock In {\em Proceedings of CAV 2011}, LNCS, vol. 6806, pages 585--591,
  2011.

\bibitem{magnani2007bpmn}
M.~Magnani and D.~Montesi.
\newblock {BPMN}: How much does it cost? {A}n incremental approach.
\newblock In {\em Business Process Management}, LNCS, vol. 4714, pages 80--87.
  Springer, 2007.

\bibitem{mazurkiewicz1986trace}
A.~Mazurkiewicz.
\newblock Trace theory.
\newblock In {\em Petri nets: applications and relationships to other models of
  concurrency}, pages 278--324. Springer, 1986.

\bibitem{puterman2014markov}
M.~L. Puterman.
\newblock {\em Markov decision processes: discrete stochastic dynamic
  programming}.
\newblock John Wiley \& Sons, 2014.

\bibitem{DBLP:conf/icws/SaeediZS10}
K.~Saeedi, L.~Zhao, and P.~R.~F. Sampaio.
\newblock Extending {BPMN} for supporting customer-facing service quality
  requirements.
\newblock In {\em {ICWS} 2010}, pages 616--623. {IEEE} Computer Society, 2010.

\bibitem{sampath2011evaluation}
P.~Sampath and M.~Wirsing.
\newblock Evaluation of cost based best practices in business processes.
\newblock In {\em Enterprise, Business-Process and Information Systems
  Modeling}, Lecture Notes in Business Information Processing, vol. 81, pages
  61--74. Springer, 2011.

\bibitem{VN}
D.~Varacca and M.~Nielsen.
\newblock Probabilistic {P}etri nets and {M}azurkiewicz equivalence.
\newblock 2003.
\newblock Unpublished Manuscript. Available online at
  \url{http://www.lacl.fr/~dvaracca/works.html}. Last retrieved on May 27,
  2016.

\bibitem{DBLP:journals/tcs/VaraccaVW06}
D.~Varacca, H.~V{\"{o}}lzer, and G.~Winskel.
\newblock Probabilistic event structures and domains.
\newblock {\em Theor. Comput. Sci.}, 358(2-3):173--199, 2006.

\end{thebibliography}


\end{document}